\documentclass{article}
\usepackage[utf8]{inputenc}

\usepackage{amsmath,amsfonts,amssymb,amsthm,mathtools,dsfont,bbold,bbm,graphicx,fancybox,comment,xcolor,caption,subcaption,booktabs,hyperref,cleveref,tikz, float}

\usepackage{arxiv}

\newtheorem{theorem}{Theorem}[section]
\newtheorem{lemma}{Lemma}[section]

\newtheorem{corollary}{Corollary}[section]
\newtheorem{remark}{Remark}[section]

\newcommand{\RR}{\mathbb{R}}

\newcommand{\bd}{\mathbf{d}}
\newcommand{\bx}{\mathbf{x}}
\newcommand{\by}{\mathbf{y}}
\newcommand{\bu}{\mathbf{u}}
\newcommand{\bv}{\mathbf{v}}
\newcommand{\br}{\mathbf{r}}
\newcommand{\bone}{\mathbf{1}}

\newcommand{\bkappa}{\mbox{\boldmath$\kappa$}}
\newcommand{\bDelta}{\mbox{\boldmath$\Delta$}}

\definecolor{brilliantrose}{rgb}{1.0, 0.33, 0.64}
\definecolor{amber}{rgb}{1.0, 0.75, 0.0}
\definecolor{amethyst}{rgb}{0.6, 0.4, 0.8}
\definecolor{carrotorange}{rgb}{0.93, 0.57, 0.13}
\definecolor{rosewood}{rgb}{0.4, 0, 0.04}
\definecolor{lincolngreen}{rgb}{0.11, 0.35, 0.02}
\definecolor{dukeblue}{rgb}{0, 0, 0.61}

\title{
Generalized Friendship Paradoxes in Network Science
}

\author{
    \textbf{Desmond J. Higham} \\
    School of Mathematics\\
    University of Edinburgh\\
    Edinburgh EH93FD, UK\\
    \texttt{d.j.higham@ed.ac.uk} 
    \and \textbf{Francesco Hrobat} \\ 
Mathematical Institute\\
University of Oxford\\
Oxford
OX2 6GG\\
\texttt{francesco.hrobat@maths.ox.ac.uk}
    \and
    \textbf{Francesco Tudisco}\\
 School of Mathematics\\
    University of Edinburgh\\
    Edinburgh EH93FD, UK\\
  \texttt{f.tudisco@ed.ac.uk} 
}

\date{May 2023, this version Oct 2024}

\begin{document}

\maketitle

\begin{abstract} 
Generalized friendship paradoxes occur when, on average, our friends have more of some attribute than us.
These paradoxes are relevant to   
many aspects of human interaction, notably in social science and epidemiology. 
Here, we  derive new theoretical results concerning the inevitability of a paradox arising, using a linear algebra perspective.
Following the seminal 1991 work of Scott L.\ Feld,
we consider two
distinct ways to measure and compare averages, which may be regarded as global and local.
For global averaging, we show that a
generalized friendship paradox holds for a 
large family of walk-based centralities, including Katz centrality and total subgraph communicability,
and also 
for 
nonbacktracking eigenvector centrality.
However, we also find counterexamples for
centralities based on walks of even length.
For local averaging we establish a paradox for 
nonbacktracking eigenvector centrality and  
we characterize the cases where the paradox holds with equality for the walk-based case.
Defining loneliness as the reciprocal of the number of friends, we show that for this attribute
the generalized and local friendship paradoxes always hold in reverse. In this sense, we are always more lonely, on average, than our friends. 
We also derive global and local averaging paradoxes for the case where the arithmetic mean is replaced by the geometric mean.
As well as unifying and adding to the literature in this area, we highlight some open questions.
\end{abstract}

\maketitle

\section{Motivation}
\label{sec:mot}

In 1991, 
the sociologist Scott L. Feld 
pointed out a sampling bias phenomenon that applies to any unweighted, undirected network \cite{Feld91}.
In the context of social networks, \emph{on average, our friends have more friends than we do}.
This effect has subsequently become known as the \emph{Friendship Paradox}.
With regard to our online interactions, 
it has been argued that the effect 
distorts our perceptions \cite{Ja19} and contributes to our dissatisfaction \cite{Bollen2017,DDG17}.
On a more positive note, the paradox may be exploited in the design of
strategies to monitor or control the spread of information or disease \cite{BAN22,CF10,CEU21,GMCCF14,Gomez7738,KKF24,KS22,PMG17},
including vaccination programmes. 
More generally, it may be used to 
sample more efficiently from the tail of a degree distribution \cite{NK21}.

In \cite{EJ14} Eom and Jo looked at a \emph{Generalized Friendship Paradox}: given an externally derived attribute, 
\emph{on average, do our neighbours have more of this attribute than we do?} For example, in a scientific collaboration network, do our 
coauthors have a higher citation count than us on average?
Such issues have subsequently been studied in a range of empirical tests, 
leading to concepts such as the 
activity and virality paradoxes \cite{Hodas13icwsm},
the H-index paradox \cite{BLA16}, 
the 
happiness paradox \cite{PMG17}
and 
the enmity paradox \cite{GC23}.
A further direction of research focuses attention on specific classes of networks, as in 
\cite{CKN21,EK21,JLY21,JLY22,SMG21}.

The original work of Feld 
\cite{Feld91} set up two distinct types of averaging process.
Following 
\cite{GC23} we will distinguish between them by using the 
terms global and local.
(We note that the terms ``ego'' and ``alter'' are used in 
\cite{KKF24}.)
Among the empirical tests for paradox effects across various application domains, some emphasize local averaging 
\cite{BLA16,CALS17,Bollen2017,CKN21,Kooti14icwsm,MR16,YAWC22,ZLZHY20} 
and others global averaging 
\cite{Gr14,Hodas13icwsm,NKNK24}.
In this work we contribute to the understanding
of both types of generalized friendship paradox by highlighting 
new cases where a paradox is guaranteed to hold for any network.
We thereby add to the 
theoretical inevitability results on global averaging in 
\cite{Feld91,H19} and on local averaging in 
\cite{CKN21,EJ14,HV25}.
Our main contributions are:
\begin{itemize}
   \item Theorem~\ref{thm:globfA} and Corollary~\ref{cor:fA} show that a global 
     averaging  paradox holds for a wide class of walk-based centrality measures. 
\item Theorems~\ref{thm:globnbt} and 
\ref{thm:locnbt}
show that both global and local  
     averaging  paradoxes hold for a nonbacktracking version of eigenvector centrality.  
   \item Theorems~\ref{thm:revgfp} and 
    \ref{thm:revlfp} show that both global and local averaging  paradoxes hold for loneliness (defined as the reciprocal of degree). 
   \item Theorems~\ref{thm:globgmfp} and 
    \ref{thm:locgmfp}
   show that global and local geometric mean paradoxes hold for degree. 
   \item Theorem~\ref{thm:even} characterizes  when equality holds in the local averaging  paradox for walk centrality.
   \item Counterexamples in subsection~\ref{subsec:counterggfi} show that a paradox does not hold with
    global averaging for even length walk centrality.
   \end{itemize}

   Section~\ref{sec:def} sets up the key inequalities and describes known results.
   In section~\ref{sec:ggfp}--\ref{sec:geom} we derive new paradoxes, and 
   in 
   section~\ref{sec:technical} we deal with some extra technicalities. The results of 
   computational experiments on real and synthetic networks are presented in 
   section~\ref{sec:exp}.  In 
    section~\ref{sec:summary} we
   summarize the current state of the art and highlight some open questions.

In deriving results, to minimize confusion we will use the phrase ``inequality'' for expressions that arise when we compare related quantities and reserve the phrase ``paradox'' for a result
showing that the inequality holds for all networks.

\section{Definitions and Existing Results}
\label{sec:def}
We consider undirected, unweighted, connected networks with $n$ nodes and no self-loops.
We let $A \in \RR^{n \times n}$ denote the adjacency matrix. 
So $A$ is a symmetric binary matrix with zero diagonal such that 
$a_{ij} = 1$ if nodes $i$ and $j$ are connected and 
$a_{ij} = 0$ otherwise.

We define the degree vector $\bd \in \RR^{n}$ by 
\[
\bd  = A \bone,
\]
where 
$
\bone \in \RR^{n}
$
is the vector of ones. So $d_i$ is the degree of node $i$.
In a social network context, $d_i$ measures the number of friends that node $i$ possesses.
We use $\bd^{-1}$ to denote the vector 
with $i$th element $1/d_i$; so 
$\bd^{-1}$ contains the reciprocal degrees.
We use $\bd^{-T}$ to denote the transpose of 
$\bd^{-1}$. Hence, when $\bx \in \RR^{n}$ we have 
 $\bd^{-T}  \bx = \sum_{i=1}^{n} x_i/d_i$. 
We let $D \in \RR^{n \times n}$ denote the 
diagonal degree matrix, for which $D_{ii} = d_i$. 

\subsection{Global Averages} \label{subsec:glob}
By definition, the (arithmetic) average degree, or average number of friends, taken over all nodes is given by 
$\bone^T \bd/n$.
Suppose instead that we consider all nodes and compute the average number of friends-of-friends. In this computation,
every node appears as a friend $d_i$ times, and each time contributes its $d_i$ friends to the sum, giving a total of $\bd^T \bd$. The number of terms in this sum is 
$\bone^T \bd$, since every edge contributes twice.
  It follows that the difference between the
  friend-of-friend average and the 
  friend average 
   is 
  \begin{equation}
  \frac{\bd^T \bd}{ \bone^T \bd} - \frac{\bone^T \bd}{n}.
\label{eq:fp1}
      \end{equation}
      Feld showed the following result, which may be proved 
      by applying standard arguments in statistics \cite{Feld91,St12} or linear algebra \cite{H19}.
      \begin{theorem}[Global friendship paradox \cite{Feld91}]\label{thm:FPglobal}
       The inequality 
   \begin{equation}
  \frac{\bd^T \bd}{ \bone^T \bd} - \frac{\bone^T \bd}{n} \ge 0 
\label{eq:fpglobal}
      \end{equation} 
      always holds, with equality if and only if the network is regular.
      \end{theorem}  
      Theorem~\ref{thm:FPglobal} may interpreted as saying that, 
      \emph{on average
      our friends have at least as many friends as us, with equality
      occuring only if we all have the same number of friends}.
      Because this inequality deals with a node-wide average, and to distinguish 
      it from (\ref{eq:fplocal}) below, we will refer to (\ref{eq:fpglobal}) as the 
      \emph{global friendship inequality}.

Eom and Jo \cite{EJ14} considered the case where each node $i$ in the graph has an associated
attribute, $x_i > 0$. 
An appropriate generalisation of the inequality 
(\ref{eq:fpglobal})
is then 
   \begin{equation}
  \frac{\bd^T \bx}{ \bone^T \bd} - \frac{\bone^T \bx}{n} \ge 0. 
\label{eq:genfpglobal}
      \end{equation} 
     We will refer to (\ref{eq:genfpglobal}) as a 
     \emph{generalized global friendship inequality}; if this holds then, on average, our neighbours have more of the attribute $\bx$ than we do.
      We emphasize that in the general case where the attribute vector $\bx$ is  independent of the network structure, there is no way of knowing a priori 
      whether the inequality (\ref{eq:genfpglobal}) will hold.
      However, there is a simple and intuitive characterisation:
      Eom and Jo \cite{EJ14} showed that the inequality is equivalent to the     
     statement that $\bx$ is nonnegatively correlated with the degree vector.
  (We say that the vectors $\bu$ and $\bv$
  in $\RR^{n}$ 
  are 
  nonnegatively
correlated if 
$\sum_{i=1}^{n}
\sum_{j=1}^{n} (u_i - u_j) (v_i - v_j) \ge 0
$,
and positively 
correlated
if the inequality holds strictly.)

    The work of Higham \cite{H19} was the first to consider the case where the attribute $\bx$ is a network centrality measure (other than degree), and hence is directly dependent on the network structure.
  According to classical 
  \emph{eigenvector  
  centrality} 
  \cite{B1,B2,Vigna}, the importance of  node $i$ can be measured by $x_i$, where 
  $\bx$ is the Perron-Frobenius eigenvector of $A$.
  Here, 
  $\bx$ has positive elements,  
  $A \bx$ is proportional to 
    $\bx$, and we
   assume without loss of generality the normalization 
  $\bone^T \bx = 1$.
  We then have the following result.
   \begin{theorem}[Global eigenvector centrality paradox \cite{H19}]\label{thm:eigFPglobal} 
       The generalized global friendship inequality (\ref{eq:genfpglobal}) holds for eigenvector centrality, 
       with equality      
      if and only if the network is regular. 
      \end{theorem}  
      Theorem~\ref{thm:eigFPglobal} may be interpreted as saying that, 
      \emph{on average, our friends are at least as important as us, with equality
      occurring only if we all have the same number of friends}. Results of this type add further weight to the arguments proposed in the  social sciences that friendship paradox effects contribute to the perceived levels of dissatisfaction among those who engage in 
      online interaction. They also support the empirical findings that the paradox can be exploited for monitoring or controlling the spread of information or disease.

  \subsection{Local Averages} \label{subsec:loc} 
  The original article \cite{Feld91} mentioned an alternative, more localized, 
  way to compare a node with its neighbours.
  For node $i$, we may 
   record how the nodal degree $d_i$ compares with the average 
    degree of its neighbours by computing 
  \[
     \Delta_i =  
     \frac{1}{d_i}
  \sum_{j=1}^{n}
    a_{ij} d_j  - d_i. 
    \]
    After making these local comparisons, we may then move to a graph-wise
    summary by
     asking whether the arithmetic average of the $\Delta_i$ values is nonnegative.
     Since $\bDelta = D^{-1} A \bd - \bd$, we may write this inequality as     
     \begin{equation}
      \bone^T ( D^{-1} A \bd - \bd ) \ge 0,
     \label{eq:fplocal}
      \end{equation}
      giving another sense in which, on average,
      our friends might have at least as many friends as us.
      We will refer to (\ref{eq:fplocal}) as 
      the \emph{local friendship inequality}.  
      In \cite{Feld91} Feld defined this concept and reported 
       results for specific networks; however, unlike in the 
       global case (\ref{eq:fpglobal}), he did not study it in general.
       Thirty years later, Cantwell et al.\ \cite{CKN21} proved the analogue of 
       Theorem~\ref{thm:FPglobal}:
        \begin{theorem}[Local friendship paradox \cite{CKN21}]\label{thm:FPlocal} 
       The inequality (\ref{eq:fplocal})  
      always holds, with equality if and only if the network is regular.
      \end{theorem}  
       
   Cantwell  et al.\ \cite{CKN21} also introduced the generalized version of  
   (\ref{eq:fplocal}) where degree is replaced by an attribute vector $\bx$.
   Here, instead of the global version (\ref{eq:genfpglobal}), we have 
    \begin{equation}
      \bone^T ( D^{-1} A \bx - \bx ) \ge 0,
     \label{eq:genfplocal}
      \end{equation}
   which we refer to as a 
     \emph{generalized local friendship inequality}.
     The authors in \cite{CKN21} showed that (\ref{eq:genfplocal}) is equivalent to the 
     statement that $\bx$ correlates 
     nonnegatively with $\bkappa$, where 
     \begin{equation}
         \kappa_i = \sum_{j=1}^{n} \frac{a_{ij}}{d_j}.
         \label{eq:kappa}
         \end{equation}

     As indicated in the theorem below, 
   Hazra and Verbitskiy \cite{HV25} showed that the generalized local friendship inequality holds for eigenvector centrality and also for 
   \emph{walk centrality}, defined as
    \begin{equation}
       \bx = A^{\ell} \bone, \quad \text{~where~}
       \ell >1 \text{~is~an~integer}, 
       \label{eq:walk}
    \end{equation} 
    and for \emph{Katz centrality} \cite{Ka}, defined as
   \begin{equation}
   (I - \alpha) \bx = \bone,
   \label{eq:Katz}
   \end{equation}
   where $0  < \alpha < 1/\rho(A)$, with 
   $\rho(A)$ denoting the spectral radius of $A$.

   \begin{theorem}[Generalized local friendship paradox for eigenvector, walk and Katz centralities \cite{HV25}] \label{thm:HV25}
   The generalized local friendship inequality 
   (\ref{eq:genfplocal}) holds for eigenvector, walk and Katz centrality.
  \end{theorem}

 The authors in \cite{HV25} 
 also discuss conditions under which equality holds 
 in 
 (\ref{eq:genfplocal}).
In the case of eigenvector and Katz 
 centrality, they prove that 
 equality holds if and only if the graph is regular.
 For walk centrality, they point out that 
 equality always holds in the degenerate case where
 $\ell = 0$ in (\ref{eq:walk}), so $\bx = \bone$.
 However, equality is not discussed for $\ell \ge 1$.
 We return to this issue in  section~\ref{sec:technical}, where we show in Theorem~\ref{thm:even} that for
 odd $\ell$ equality holds if and only if the graph is regular and 
 for 
 even $\ell$ 
  equality holds if and only if the graph is either regular or biregular.

In the following sections we
demonstrate that further global and local paradoxes may be established.

\begin{remark}\label{rem:scaling}
    We note that the inequalities
    (\ref{eq:genfpglobal}) and (\ref{eq:genfplocal}) have two natural properties. First, if 
    $\bx$ satisfies an inequality
    then so does $\beta \bx$ for any $\beta > 0$.
    Second, if $\bx$ and $\by$ satisfy an inequality, then so does $\bx + \by$.
\end{remark}

    \section{
      Generalized Global Friendship Paradoxes}
    \label{sec:ggfp}

    In this section we 
     focus on the generalized global friendship inequality (\ref{eq:genfpglobal}) and derive new results
     concerning centrality measures. 

      \subsection{Generalized Global Friendship Paradox for 
      Katz-style Centralities}\label{subsec:ggfpkatz}

      We begin with a wide class of matrix-function based centralities; see, for example, \cite{EHSiamRev} for further motivation.

     \begin{theorem}[Generalized global friendship paradox for $f(A)\bone$] \label{thm:globfA}
 The generalized global friendship inequality (\ref{eq:genfpglobal}) holds for 
$\bx = f(A)\bone$, where $f$ is an increasing function defined over the spectrum of $A$.
Furthermore, if $f$ is strictly increasing then equality holds if and only if the associated graph is regular.
  \end{theorem}
  \begin{proof}
Since $A$ is real and symmetric, it admits an orthonormal eigenbasis $\{\bv_j\}$ with
corresponding real eigenvalues $\lambda_j$.
So we may write (with all sums running from $1$ to $n$)
\begin{equation}
A=\sum_j \lambda_j \bv_j \bv_j^T,\qquad 
\bone =\sum_j (\bv_j^T \bone)\,\bv_j,
\label{eq:sym1}
\end{equation}
and also 
\begin{equation}
\bx=f(A)\bone =\sum_j f(\lambda_j)(\bv_j^T \bone) \bv_j,
\qquad
\bd=A\bone=\sum_j \lambda_j (\bv_j^T \bone)\,\bv_j.
\label{eq:sym2}
\end{equation}
Letting $w_j :=(\bv_j^T \bone)^2\ge0$ it follows that 
\[
\bone^T \bd=\sum_j \lambda_j w_j,\qquad
\bd^T \bx=\sum_j \lambda_j f(\lambda_j)w_j,\qquad
\bone^T \bx=\sum_j f(\lambda_j)w_j.
\]
The inequality (\ref{eq:genfpglobal})
is therefore equivalent to
\begin{equation}
n\sum_j\lambda_jf(\lambda_j) w_j
\;\ge\;
\Big(\sum_j\lambda_j w_j\Big)\Big(\sum_jf(\lambda_j)w_j\Big).
\label{eq:weq}
\end{equation}

Now define the quantity 
\[
\mu:=\frac{1}{n}\sum_j \lambda_j w_j = \frac{\bone^T \bd}{n},
\]
so that the desired inequality (\ref{eq:weq}) may be written 
\begin{equation}\label{eq:rewineq}
\sum_j w_jf(\lambda_j)(\lambda_j-\mu)\ge 0.
\end{equation}

To make progress, we will establish the identity 
\begin{equation}
 \sum_j w_jf(\lambda_j)(\lambda_j-\mu) = \frac{1}{2n}\sum_{i,j}w_iw_j(\lambda_i-\lambda_j)(f(\lambda_i)-f(\lambda_j)).
\label{eq:wident}
\end{equation}

Using the fact that
\[ 
\sum_{i} w_i = \sum_{i} (\bv_i^T \bone)^2 = \| \bone \|_2^2  = n,
\]
we have 
\begin{eqnarray*}
\frac{1}{2n}\sum_{i,j}w_iw_j(\lambda_i-\lambda_j)(f(\lambda_i)-f(\lambda_j))&=& 
\frac{1}{2n}\sum_{i,j} w_i w_j\Big(\lambda_i f(\lambda_i)-\lambda_i f(\lambda_j)-\lambda_j f(\lambda_i)+\lambda_j f(\lambda_j)\big)\\
&=&\frac{1}{2n}\Big(\Big(\sum_jw_j\Big)\Big(\sum_i w_i\lambda_i f(\lambda_i)\Big) + \Big(\sum_iw_i\Big)\Big(\sum_j w_j\lambda_j f(\lambda_j)\Big) \\
&& \mbox{} 
- \sum_{i,j} w_i w_j \lambda_i f(\lambda_j) - \sum_{i,j} w_i w_j \lambda_j f(\lambda_i)\Big)\\
&=&\frac{1}{2n}\Big(2n\sum_i w_i\lambda_i f(\lambda_i)
- 2 \sum_{i,j} w_i w_j \lambda_i f(\lambda_j) \Big),\\
&=& \sum_i w_i\lambda_i f(\lambda_i)
- \frac{1}{n} 
\Big(\sum_i w_i \lambda_i\Big)\Big(\sum_j w_j f(\lambda_j)\Big),\\
&=& \sum_i w_i\lambda_i f(\lambda_i)
-  \mu \sum_j w_j f(\lambda_j),\\
&=& \sum_i w_i f(\lambda_i) (\lambda_j - \mu),
\end{eqnarray*}
as required. 

Because $f$ is increasing we have 
$(\lambda_i - \lambda_j)(f(\lambda_i)-f(\lambda_j)) \ge 0$ for any $i$ and $j$. Summing, we see that the expression on the right hand side of (\ref{eq:wident})
is greater than or equal to zero, and hence (\ref{eq:rewineq}) holds.

If $f$ is strictly increasing then it is clear from the identity 
(\ref{eq:wident}) that 
equality holds 
 and only if
 all $\lambda_i$ are equal
 for eigenvalues whose eigenvector is orthogonal to $\bone$.
  In this case, from (\ref{eq:sym1}) and (\ref{eq:sym2}), we have 
\[
\bd = \sum_{j} \lambda_j(\bv_j^T \bone) \bv_j
= \lambda \sum_{j} (\bv_j^T \bone)  \bv_i
= \lambda \bone,
\]
so the graph is regular.
On the other hand, if the graph is regular, then 
\[ \frac{\bd^T \bx}{\bone^T \bd} = \frac{\lambda \bone^T \bx}{\lambda \bone^T \bone} = \frac{\|\bx\|_1}{n},\]
completing the proof.
  \end{proof}

  Theorem~\ref{thm:globfA} immediately covers 
  walk centrality (\ref{eq:walk}), 
  Katz centrality (\ref{eq:Katz}) and 
   \emph{total subgraph communicability} \cite{BK13},
    also known as 
  \emph{exponential centrality}, which is 
  defined by 
  \begin{equation}
       \bx = \exp(A) \bone.
       \label{eq:total}
    \end{equation}

\begin{corollary}[Examples arising from $f(A)\bone$]
\label{cor:fA}
The generalized global friendship inequality (\ref{eq:genfpglobal}) holds for 
walk centrality (\ref{eq:walk}) when $\ell$ is odd,
Katz centrality (\ref{eq:Katz}) and total subgraph communicability (\ref{eq:total}).
Furthermore, in these cases 
equality holds 
if and only if the associated graph is regular.

\end{corollary}

Note that the generalized global friendship inequality does \textit{not} hold in general for the centrality measure given by walk centrality when $\ell$ is even. We discuss this issue in section~\ref{sec:technical}.

   \subsection{Generalized Global Friendship Paradox for 
      Nonbacktracking Eigenvector Centrality}\label{subsec:ggfpnb}
      For some networks, the standard eigenvector centrality measure has been observed
to exhibit \emph{localization}---attributing large values to a small subset of very central nodes, and failing to distinguish 
adequately between the remainder \cite{AGHN17a,KMMNSZZ13,MZN14,PC15}.
To remedy this, 
an alternative \emph{nonbacktracking} version of eigenvector
centrality 
was proposed in \cite{MZN14}.

    The nonbacktracking eigenvector centrality vector 
    $ \bv_a \in \RR^{n}$ can be computed via 
the $2n$ by $2n$ eigenvalue problem
\begin{equation}
\left[
\begin{array}{cc}
A & I - D \\
I & 0 
\end{array}
\right]
\left[
\begin{array}{c}
\bv_a \\
\bv_b
\end{array}
\right]
=
\lambda
\left[
\begin{array}{c}
\bv_a \\
\bv_b
\end{array}
\right],
\label{eq:nbteig}
\end{equation}
which is intimately connected with the nonbacktracking, or
Hashimoto, matrix \cite{MZN14,NQ24,ST96}
and also with
 nonbacktracking Katz-style centrality \cite[Theorem~10.2]{GHN18}. 
Generically, the $2n$ by $2n$ matrix in (\ref{eq:nbteig}) has a spectral radius strictly greater than one, with a dominant real-valued 
eigenvalue $\lambda > 1$ for which 
$\bv_a$ and
$\bv_b$ in 
the 
corresponding eigenvector can be chosen to have nonnegative entries.
Following the authors in \cite{MZN14}, we then 
use $\bv_a$ to define the nonbacktracking eigenvector centrality measure, normalized so that 
$\bone^T \bv_a = 1$.
However, for certain graphs the spectral radius is equal to one; this is a case where the dominant eigenvalue may not be not unique.
It is immediate in (\ref{eq:nbteig}) that 
$\lambda = 1$ is an eigenvalue with corresponding eigenvector given by $\bv_a = \bv_b = \bone$. Hence when $\lambda = 1$ 
is a dominant eigenvalue we take the natural choice of $\bx = \bone/n$
to define the centrality measure. In this scenario, 
where the centrality measure is uniform, both versions of the generalized friendship inequality hold with equality. So it is of interest to characterize when $\lambda = 1$ is a dominant eigenvalue.
This motivates the material in 
subsection~\ref{subsec:l1}.
For completeness, 
we also provide arguments in 
subsection~\ref{subsec:nnbtheory}
which confirm that the centrality measure is
always well defined. 

We now show that a generalized global friendship paradox holds for the nonbacktracking eigenvalue centrality measure.

 \begin{theorem}[Generalized global nonbacktracking eigenvector
  centrality paradox] \label{thm:globnbt}
 The generalized global friendship inequality (\ref{eq:genfpglobal}) holds for 
 nonbacktracking eigenvector
  centrality (defined via $\bx = \bv_a$ in (\ref{eq:nbteig})). 
  \end{theorem}
  \begin{proof}
  Let $\lambda \ge 1$ denote the largest real-valued eigenvalue of the system (\ref{eq:nbteig}). 
  As mentioned above, if $\lambda = 1$ then 
  $\bv_a = \bone/n$ and the inequality holds with equality. 

  Hence we may assume that there is a unique 
  real dominant eigenvalue $\lambda >1$.

  From (\ref{eq:nbteig}), we have 
  \begin{equation}
  A \bv_a = \lambda \bv_a + (D-I) \bv_b,
   \label{eq:nba}
  \end{equation}
  and
  \[
   \bv_a = \lambda \bv_b.
   \]
  So
  \[
    A \bv_a = \lambda \bv_a + (D-I) \frac{\bv_a}{\lambda},
  \]
  which may be written 
    \[
    \lambda A \bv_a = ((\lambda^2 - 1) I  + D)  \bv_a.
  \]
  So
    \[
    \lambda \bone^T A \bv_a = \bone^T ((\lambda^2 - 1) I  + D)  \bv_a.
  \]
  Using $\bone^T \bv_a = 1$,
  this simplifies to
    \[
    (\lambda -1 ) \bd^T \bv_a =  (\lambda^2 - 1).
  \]
  Equivalently 
     \[
    \bd^T \bv_a =  \lambda + 1.
  \]
  
So (\ref{eq:genfpglobal}) holds for $\bx = \bv_a$ if and only if
\begin{equation}
\lambda \ge \frac{\bone^T \bd}{n} -1. 
\label{eq:terras}
\end{equation}

We now make use of the connection between the eigenvalue problem (\ref{eq:nbteig}), the 
Hashimoto, or nonbacktracking,  matrix $B$ (which we define in subsection~\ref{subsec:l1}),
and the radius of convergence, $R$, of the Ihara zeta function of the graph.
The Ihara zeta function is defined as
\[ \zeta(u)^{-1} = \mathrm{det}(I - u B).\]
It is known  
that 
$\lambda$ is equal to the largest eigenvalue of the 
nonbacktracking matrix, which itself is equal to the reciprocal of $R$ (see \cite[Thm~5.6.15]{Hor2013} applied to the matrix $zA$); so $\lambda = 1/R$.
Terras \cite{T2011} conjectured that 
\[
 \frac{1}{R} \ge \frac{\bone^T \bd}{n} -1
\]
(which is precisely (\ref{eq:terras})) 
for a graph that does not have a node of degree one and is not a cycle. This conjecture was proved by Saito in 2018 \cite{S18}.

To finish we need to deal with the remaining cases
where the graph has a node of degree one or the graph is a cycle.

It is shown in \cite[Corollary~3.2.1]{GK21} that increasing the size of the graph by taking one node and joining it to a tree 
does not change $\lambda$. Hence, removing a single node of degree one, along with the corresponding edge, does not change $\lambda$.
See also \cite[Theorem~1.2]{NQ24} for a more direct argument that also applies to directed graphs.
We may therefore iteratively remove all nodes of degree one until we have a graph where all nodes have degree at least two, whence the analysis above applies if the remaining graph is not a cycle.

If the graph is a cycle, then from \cite[Proposition~3.6]{T20} $\lambda = 1$
and we return to the equality case. 
\end{proof}

    \section{
      Generalized Local Friendship Paradoxes}
    \label{sec:glfp}

We now prove the equivalent of Theorem~\ref{thm:globnbt} for the local case.

\begin{theorem}[Generalized local nonbacktracking eigenvector
  centrality paradox] \label{thm:locnbt}
 The generalized local friendship inequality (\ref{eq:genfplocal}) holds for 
 nonbacktracking eigenvector
  centrality (defined via $\bx = \bv_a$ in (\ref{eq:nbteig})). 
  \end{theorem}
  \begin{proof}
    Premultiplying (\ref{eq:nba})
    by $\bone^T$,
 using $\bone^T \bv_a = 1$
 and rearranging,
 we have 
  \begin{equation}
  \lambda = \bd^T \bv_a - \bone^T ( D - I) \bv_b.
  \label{eq:lambda}
  \end{equation}

 Starting from the left-hand side of the inequality 
 (\ref{eq:genfplocal}) with $\bx = \bv_a$, we have, using (\ref{eq:nba}),  
  \begin{eqnarray*}
      \bone^T \left( D^{-1} A \bv_a - \bv_a \right) & = & 
          \bone^T \left( D^{-1} (\lambda \bv_a + (D-I) \bv_b)
           - \bv_a \right)\\
           &=& \lambda \bd^{-T} \bv_a + \bone^T (D-I) \bv_b - 1.
  \end{eqnarray*}
  Using the expression (\ref{eq:lambda}) for $\lambda$, this gives
  \begin{eqnarray}
      \bone^T \left( D^{-1} A \bv_a - \bv_a \right) & = &
      \left(
      \bd^T \bv_a - \bone^T ( D - I) \bv_b
      \right)
      \bd^{-T} \bv_a
      +
       \bone^T (D-I) \bv_b - 1 \nonumber\\
       &=& (\bd^T \bv_a) (  \bd^{-T} \bv_a ) - 1 +  \bone^T ( D - I) \bv_b \left( 1 - \bd^{-T} \bv_a \right).
           \label{eq:nonneg}
        \end{eqnarray}
Now the 
 arithmetic mean-harmonic mean inequality 
 shows that $ (\bd^T \bv_a) (  \bd^{-T} \bv_a ) \ge 1  $.
 Because $\bone$, $D - I$ and $\bv_b$ contain nonnegative elements, we have 
 $\bone^T ( D - I) \bv_b \ge 0$.
 Finally, since $\bone^T \bv_a = 1$ and each $d_i \ge 1$, we have 
 $\bd^{-T} \bv_a  \le 1$.
 So the right-hand side of (\ref{eq:nonneg}) is nonnegative, as required.
  \end{proof}

  We note that in the statement of Theorem~\ref{thm:locnbt} we cannot add that equality holds if and only if the network is regular.
  To see this, suppose $\bd = A \bone = d \bone$ for some common degree $d \ge 2$. Then the the right hand side of
  \eqref{eq:nonneg} reduces to 
  $(d-1)(1-\frac{1}{d})\bone^T\bv_b$, which 
  is nonzero.

    \section{Generalized Friendship Paradoxes: Global versus Local} \label{sec:globloc}
The generalized global and local friendship inequalities are similar in the sense that 
both hold universally for the attributes of degree 
(Theorems~\ref{thm:FPglobal} 
and \ref{thm:FPlocal})
and eigenvector centrality (Theorems~\ref{thm:eigFPglobal} and \ref{thm:HV25}).
However, the characterizations mentioned in subsections~\ref{subsec:glob} and \ref{subsec:loc} suggest that they are not equivalent---the global version corresponds to 
correlation with the degree, while the local version corresponds to correlation with the 
vector $\bkappa$ in (\ref{eq:kappa}). Here, we confirm with a concrete example that 
all four possibilities of the two inequalities holding/not holding may arise.
We consider the four-node network with adjacency matrix
\[
A = \left[ 
\begin{array}{cccc}
 0 & 1 & 0 & 0 \\
 1 & 0 & 1 & 1 \\
 0 & 1 & 0 & 1 \\
 0 & 1 & 1 & 0 
 \end{array}
 \right].
 \]
 Given an attribute vector $\bx$ 
 it is straightforward to show that the 
 generalized global friendship inequality
 (\ref{eq:genfpglobal}) 
 and  
  generalized local friendship inequality
  (\ref{eq:genfplocal}) are 
   equivalent to  
 \[
 x_2 - x_1 \ge 0 \quad \text{and} \quad 
 -4 x_1 + 6 x_2 - x_3 - x_4 \ge 0,   
 \]
 respectively. It follows that 
 \begin{itemize}
 \item with $x_1 = 1$, $x_2 = 2$ and $x_3 = x_4 = 6$ we have (\ref{eq:genfpglobal})
 but not (\ref{eq:genfplocal}),
  \item with $x_1 = 7$, $x_2 = 6$ and $x_3 = x_4 = 2$ we have 
  (\ref{eq:genfplocal}) 
  but not
  (\ref{eq:genfpglobal}),
   \item with $x_1 = 2$, $x_2 = 1$ and $x_3 = x_4 = 1$ we have
   neither 
    (\ref{eq:genfpglobal})
    nor
  (\ref{eq:genfplocal}),
    \end{itemize}
    and (immediately from Theorems~\ref{thm:FPglobal} 
and \ref{thm:FPlocal}) with $\bx = \bd$ we have both
      (\ref{eq:genfpglobal})
    and 
  (\ref{eq:genfplocal}).

 \section{Loneliness Paradoxes} \label{sec:reverse}

  By analogy with (\ref{eq:genfpglobal}) we may say that 
  a \emph{reverse generalized global friendship inequality} arises if, on average, a node has more of an attribute than its neighbours; that is, 
   \begin{equation}
  \frac{\bd^T \bx}{ \bone^T \bd} - \frac{\bone^T \bx}{n} \le 0. 
\label{eq:genfpglobalrev}
      \end{equation} 
Here we show the intuitively reasonable result that 
the reciprocal of degree has this property. Thus if we regard
$\bd^{-1}$ as a measure of loneliness, then 
on average we are at least as lonely as our friends.

\begin{theorem}[Global loneliness paradox]\label{thm:revgfp} 
       The 
        reverse generalized global friendship  
       inequality (\ref{eq:genfpglobalrev}) holds for 
       $\bx = \bd^{-1}$,
       with equality      
      if and only if the network is regular. 
      \end{theorem}  
     \begin{proof}   
     In this case (\ref{eq:genfpglobalrev}) becomes
     \[
      \frac{\bd^T \bd^{-1}}{ \bone^T \bd} - \frac{\bone^T \bd^{-1}}{n} \le 0,
     \]
     which simplifies to
    \[
      \frac{n}{ \bone^T \bd} - \frac{\bone^T \bd^{-1}}{n} \le 0.
     \] 
      The result now follows directly from the 
      arithmetic mean-harmonic mean inequality \cite{MW12}. 
     \end{proof}

Similarly, we may say that a 
       \emph{reverse generalized local friendship inequality} arises if (\ref{eq:genfplocal}) is reversed, so that 
        \begin{equation}
      \bone^T ( D^{-1} A \bx - \bx ) \le 0.
     \label{eq:genfplocalrev}
      \end{equation}
      The next result shows that the corresponding paradox holds.

\begin{theorem}[Local loneliness paradox]\label{thm:revlfp} 
       The 
        reverse generalized local friendship  
       inequality (\ref{eq:genfplocalrev}) holds for 
       $\bx = \bd^{-1}$,
       with equality      
      if and only if the network is regular. 
      \end{theorem}  
     \begin{proof}  
     We use similar arguments to those in 
     \cite[Section~2]{CKN21}.
 Taking $ x_i = 1/d_i$, in the left hand side of 
 (\ref{eq:genfplocalrev}) we have 
\[ \bone^T(D^{-1}A \bx- \bx) = \sum_{i,j} \left(\frac{a_{ij}}{d_id_j}\right) -\left(\sum_i\frac{1}{d_i}\right) = \sum_{i,j}\left(\frac{a_{ij}}{d_id_j} -\frac{a_{ij}}{d_i^2}\right) = \sum_{i,j}\left(a_{ij}\left(\frac{1}{d_id_j} -\frac{1}{d_i^2}\right)\right).\]
Since $A$ is symmetric we can rewrite this expression as
\[\frac{1}{2}\sum_{i,j}\left(a_{i,j}\left(\frac{2}{d_id_j} -\frac{1}{d_i^2} -\frac{1}{d_j^2}\right)\right) = -\frac{1}{2}\sum_{i,j}\left(a_{i,j}\left(\frac{1}{d_i} - \frac{1}{d_j}\right)^2\right) \le 0,\]
giving the required inequality.
We see that equality requires
$d_i = d_j$ whenever nodes $i$ and $j$ share an edge. Since the graph is connected, this requires the graph to be regular.
 \end{proof}

 \section{Geometric Mean Friendship Paradoxes} \label{sec:geom}

    The aim of this section is to 
consider using the geometric mean
in 
(\ref{eq:fpglobal})
and 
(\ref{eq:fplocal})
rather than the 
arithmetic mean.

\subsection{Global Geometric Mean Paradox}\label{subsec:globgeo}

To study the global case, we will make use of 
   the following lemma.

      \begin{lemma}\label{lem:mon} 
       Given a strictly increasing function $f :\RR \to \RR$,  
       the generalized global friendship inequality (\ref{eq:genfpglobal}) 
       holds 
       for $x_i = f(d_i)$, with equality if and only if the network is regular. 
       \end{lemma}

\begin{proof}
Since $f$ is strictly increasing, we have, for any $i$ and $j$, 
$
\left(d_i - d_j\right) \left( f(d_i) - f(d_j) \right) \ge 0$,
with equality if and only if $d_i = d_j$.
Summing, we obtain,
\[
\sum_{i,j}
\left(d_i - d_j\right) \left( f(d_i) - f(d_j) \right) \ge 0,
\]
with equality if and only if the network is regular.
Hence, $f(\bd)$ is nonnegatively correlated with $\bd$, and the correlation is positive
unless the graph is regular.

Upon simplification and rearrangement, this inequality reduces to 
\[
\sum_{i} \frac{d_i f(d_i)}{\bone^T \bd}
-
\frac{1}{n}
\sum_{i} f(d_i) \ge 0,
\]
as required.
\end{proof}

As we mentioned in section~\ref{sec:def},
Eom and Jo \cite{EJ14} showed that 
(\ref{eq:genfpglobal})
is equivalent to asking for $\bx$ to be  nonnegatively correlated with the degree vector.
Lemma~\ref{lem:mon} shows the intuitively reasonable result that applying an increasing
function to the degree vector respects the inequality.

As mentioned in the start of subsection~\ref{subsec:glob},
 the friend-of-friend list includes each degree value $d_i$ exactly $d_i$ times, and the list involves $\bone^T \bd$
numbers.
So the friend-of-friend geometric mean is 
\[
\left(\prod_{i=1}^{n}  {d_i}^{d_i}\right)^{1/\bone^T \bd}.
\]
The friend geometric mean is 
\[
\left(\prod_{i=1}^{n}  {d_i} \right)^{1/n}.
\]
So we may define the \emph{global geometric mean friendship inequality} to be 
\begin{equation}
\left(\prod_{i=1}^{n}  {d_i}^{d_i}\right)^{1/\bone^T \bd} \ge \left(\prod_{i=1}^{n}  {d_i} \right)^{1/n}.
\label{eq:globgeomfp}
\end{equation}

\begin{theorem}[Global geometric mean friendship paradox]\label{thm:globgmfp} 
       The 
       global geometric mean friendship inequality
       (\ref{eq:globgeomfp})
       holds, 
       with equality      
      if and only if the network is regular. 
      \end{theorem}  
     \begin{proof}   
Taking logarithms in (\ref{eq:globgeomfp}) 
and simplifying gives  
\[
\sum_{i} 
\frac{d_i \log {d_i}}{ \bone^T \bd   }
   \ge 
\frac{1}{ n  }
\sum_{i} 
\log {d_i}.
\]
The result now follows from Lemma~\ref{lem:mon} with $f(\bd) = \log(\bd)$.



\end{proof}

\subsection{Local Geometric Mean Paradox}\label{subsec:locgeo}

To study the local case, we introduce an analogue of Lemma~\ref{lem:mon}.
This shows that applying an increasing function 
to the degree vector also respects the local inequality.

      \begin{lemma}\label{lem:mon2} 
       Given a strictly increasing function $f :\RR \to \RR$,  
       the generalized local friendship inequality (\ref{eq:genfplocal}) 
       holds 
       for $x_i = f(d_i)$, with equality if and only if the network is regular. 
       \end{lemma}

\begin{proof}
With $\bx = f(\bd)$ in (\ref{eq:genfplocal}) we have 
\begin{eqnarray*}
\bone^T ( D^{-1} A f(\bd) - f(\bd) )
  &=&
  \sum_{i,j}a_{ij}\frac{f(d_j)}{d_i} - \sum_j f(d_j)\\
  &=& \sum_{i,j}a_{ij}\frac{f(d_j)}{d_i} - \sum_{i,j}a_{ij}\frac{f(d_j)}{d_j}\\
  &=& \sum_{i,j}a_{ij}\left(\frac{f(d_j)}{d_i}-\frac{f(d_j)}{d_j}\right).
  \end{eqnarray*}
By symmetrizing we may then write
\begin{equation}
\bone^T ( D^{-1} A f(\bd) - f(\bd) )
=
\frac{1}{2}\sum_{i,j}a_{ij}\left(\frac{f(d_j)}{d_i}+\frac{f(d_i)}{d_j}-\frac{f(d_j)}{d_j}-\frac{f(d_i)}{d_j}\right).
\label{eq:feq}
\end{equation}

Now, consider each term in the summation on the right hand side of (\ref{eq:feq}). Assuming without loss of generality that 
$d_j \ge d_i$ we have 
\[ f(d_j) \geq f(d_i) \quad \mathrm{and}
\quad 
\frac{1}{d_i} \geq \frac{1}{d_j}.
\]
So, by an application of the rearrangement inequality
\[
\frac{f(d_j)}{d_i}+\frac{f(d_i)}{d_j}-\frac{f(d_j)}{d_j}-\frac{f(d_i)}{d_j}
\ge 0,
\]
with equality if and only if $d_i = d_j$.
The result follows.
\end{proof}

For a geometric mean analogue of the local inequality 
(\ref{eq:fplocal}) is natural 
to compare the geometric mean of the 
``geometric mean degree of node neighbours'' with the 
``geometric mean degree of the nodes.''
So we may define the \emph{local geometric mean friendship inequality} to be 
\begin{equation}
\left( 
\prod_i
\left(
\prod_j 
d_j^{a_{ij}}
\right)^{1/d_i}
\right)^{1/n}
\ge
\left( 
\prod_i
d_i
\right)^{1/n}.
\label{eq:locgeomfp}
\end{equation}

\begin{theorem}[Local geometric mean friendship paradox]\label{thm:locgmfp} 
       The 
       local geometric mean friendship inequality
       (\ref{eq:locgeomfp})
       holds, 
       with equality      
      if and only if the network is regular. 
      \end{theorem}  
     \begin{proof}
     After taking logarithms, 
     (\ref{eq:locgeomfp})
     simplifies to 
\[
\sum_{i}  \frac{1}{d_i} \sum_j a_{ij}  \log(d_j)  
\ge
\sum_i
\log(d_i). 
\]

This may be written
\[
\bone^T \left( D^{-1} A \log(\bd) - \log(\bd) \right) \ge 0.
\]
The result now follows from Lemma~\ref{lem:mon2} with $f(\bd) = \log(\bd)$.
\end{proof}

\section{Further Technicalities}\label{sec:technical}

In this section we tie up some loose ends arising from the previous sections, and briefly discuss some extensions.

\subsection{Equality in the Generalized Local Friendship Inequality with Even Walk Centrality}
\label{subsec:eqglfi}

We now consider when 
the generalized local friendship inequality 
(\ref{eq:genfplocal}) holds with equality
for walk count centrality
(\ref{eq:walk}).
The authors in \cite[Section~3.3]{HV25}
only mention the trivial case $\ell = 0$, where
equality always holds.
For general $\ell >0$, the situation is slightly more delicate.
For example, consider a (non-regular) complete bipartite graph based on sets of two and three nodes.
In this case,
\[
A = 
\left[
\begin{array}{ccccc}
0 & 0 & 1 & 1 & 1\\
0 & 0 & 1 & 1 & 1\\
1 & 1 & 0 & 0 & 0\\
1 & 1 & 0 & 0 & 0\\
1 & 1 & 0 & 0 & 0
\end{array}
\right],
\quad
\bd = 
\left[
\begin{array}{c}
3 \\
3 \\
2 \\
2 \\
2
\end{array}
\right],
\quad
A^2 = 
\left[
\begin{array}{ccccc}
3 & 3 & 0 & 0 & 0\\
3 & 3 & 0 & 0 & 0\\
0 & 0 & 2 & 2 & 2\\
0 & 0 & 2 & 2 & 2\\
0 & 0 & 2 & 2 & 2
\end{array}
\right],
\quad
\bx = A^2 \bone = 
\left[
\begin{array}{c}
6 \\
6 \\
6 \\
6 \\
6
\end{array}
\right].
\]
So
\[
D^{-1} A \bx = 
\left[
\begin{array}{c}
6 \\
6 \\
6 \\
6 \\
6
\end{array}
\right].
\]
Hence we have equality in (\ref{eq:genfplocal})
when $\ell = 2$. This argument readily extends to any complete bipartite graph.

The next theorem characterizes the equality case,
showing that we must distinguish between odd and even walk length. 
\begin{theorem}[Equality in the generalized local friendship paradox for walk centrality] \label{thm:even}
   For walk centrality (\ref{eq:walk}), 
    we have equality in the 
   generalized local friendship inequality 
   if and only if 
   \begin{itemize}
       \item for $\ell$ odd the graph is regular,
       \item for $\ell$ even the graph is
   either regular or biregular.
   \end{itemize}
  \end{theorem}
\begin{proof}
With $\bx = A^{\ell} \bone$, we have 
\[
 \bone^T  D^{-1} A \bx = 
 \bone^T  D^{-1} A^{\ell} D \bone = 
 \sum_{i,j} (A^{\ell}) \frac{d_i}{d_j}.
\]
So, in (\ref{eq:genfplocal}),
\begin{eqnarray*}
\bone^T ( D^{-1} A \bx - \bx )
&=& 
 \sum_{i,j} (A^{\ell})_{ij} \left(\frac{d_i}{d_j} - 1 \right)\\
 &=&
 \frac{1}{2}
 \left(
\sum_{i,j} (A^{\ell})_{ij} \left(\frac{d_i}{d_j} - 1 \right)
+
\sum_{i,j} (A^{\ell})_{ij} \left(\frac{d_j}{d_i} - 1 \right)
 \right)\\
 &=&
 \frac{1}{2}\sum_{i,j} (A^{\ell})_{ij} \left(\frac{d_i}{d_j} 
 +
 \frac{d_j}{d_i} 
 -2
\right).
\end{eqnarray*}

Now the arithmetic-geometric mean inequality tells that 
\[ 
\frac{d_j}{d_i}+\frac{d_i}{d_j} \geq 2\sqrt{\frac{d_j}{d_i}\frac{d_i}{d_j}}=2
\]
with equality if and only if $d_i = d_j$.
Hence $\bone^T ( D^{-1} A \bx - \bx ) = 0 $
if and only if 
\begin{equation}\label{eq:walkeq}
\text{for~every~} i \text{~and~} j 
\text{~we~have~either~} 
d_i = d_j \text{~or~} (A^{\ell})_{ij} = 0. 
\end{equation}
This is equivalent to the property that all  nodes in a connected component of the graph 
induced by 
$A^\ell$
have the same degree.

It follows immediately that we have equality for any regular graph.

Recall that the graph is assumed to be connected, so 
$A$ is irreducible. 
Consider first the case where $\ell$ is odd. Note that if $i$ and $j$ are connected in the original graph, then they are also connected by the walk of length 
$\ell$ given by 
\[ i -j-i-j-\cdots-i-j.\]
Hence,  $i$ and $j$ will also be 
connected in the graph induced by $A^\ell$. 
This shows that the graph induced by $A^\ell$ has at least all the edges of the original graph. Hence $A^\ell$ is irreducible, so it has only one connected component. As a consequence, (\ref{eq:walkeq}) holds if and only if the graph induced by $A^\ell$ is regular, in which case $\bone$ is an eigenvector of $A^\ell$. As $A^\ell$ and $A$ have the same eigenvectors, $\bone$ is an eigenvector of $A$ hence $A$ is regular.

Consider now the case of even $\ell$. It follows from \cite[Theorem 3.2]{powernon} 
that 
 $A$ has all powers irreducible if and only if the graph is not bipartite. In this case, we may proceed as  before, showing that  equality arises in the non-bipartite case if and only if the graph is regular. If the graph is bipartite, $A^2$ has exactly two connected components. As a consequence, for all $s \in \mathbb{N}$, $A^{2s}$ has at least $2$ connected components. To be precise, $A^{2s}$ has exactly $2$ connected components. In fact, consider $i$ and $j$ connected in $A^2$. This means that exists a node $k$ such that
 \[ i -k -j\]
 is a path in $A$. Then 
 \[ i-k-j-k-j-\dots-k-j\]
 is a path of length $2s$. 
 In conclusion, from 
 (\ref{eq:walkeq})
 we have equality if and only if the two connected components of $A^{\ell}$ are regular. Without loss of generality, we may assume that
 \[ A^2 = \begin{bmatrix}
 A_1 & 0\\
 0 & A_2
 \end{bmatrix}
 \quad
 \text{~and~consequently}
 \quad
  A^{2s} = \begin{bmatrix}
 A_1^s & 0\\
 0 & A_2^s
 \end{bmatrix}.\]
 As argued above, if $A_i^{s}$ is regular
 then $A_i$ is regular. This is precisely the definition of biregular graph.
\end{proof}

\begin{remark}\label{remark:r1}
We note that the proof above readily extends to show that the 
generalized local friendship inequality 
   (\ref{eq:genfplocal}) holds for eigenvector, and Katz centrality, 
   with equality if and only if the graph is regular. This is also shown in \cite{HV25}.
   Similarly, we may draw the same conclusion
   for exponential 
   centrality.
   \end{remark}

\begin{remark}\label{remark:r2}
    Theorem~\ref{thm:even} allows us to 
    conclude that the 
    generalized local friendship inequality
    holds for $\bx = f(A) \bone$
     when $f(A) = \sum_{k = 0}^{\infty} c_k A^k$
     takes the form of a convergent power series with all $c_k \ge 0$.
     If $c_k > 0$ for at least one odd value of $k$ then we have equality if and only if the graph is regular.
     If $c_k = 0 $ for all odd $k$ 
     and $c_k > 0$ for at least one even $k$
     then we have equality if and only if 
     the graph is either regular or biregular. 
     \end{remark}

\subsection{Counterexamples for Generalized Global Friendship Inequality with Even Walk Centrality}
\label{subsec:counterggfi}

We recall that Corollary~\ref{cor:fA} does not cover the case of even walk length centrality.
In fact when the walk length $\ell = 2$ in 
(\ref{eq:walk}) the inequality 
(\ref{eq:genfpglobal})
coincides with the difference of two Zagreb indices, which  arise in chemical graph theory 
\cite{Abd2012,Han2007}.
The authors in \cite{Han2007} studied this difference, and 
the 46 node graph displayed in 
\cite[Figure~2]{Han2007}
provides an example where
$\bx = A^2 \bone$
gives a value of 
$\approx -8 \times {10}^{-3}$
in the left hand side of 
(\ref{eq:genfpglobal}). Hence, we cannot extend the statement in Corollary~\ref{cor:fA} to 
walk length centrality with $\ell = 2$. 

A more simple counterexample 
for the $\ell = 2$ case is given by the 
19 node graph in Figure~\ref{fig:counter}.
Here, direct calculation shows that  
        $\bd^T \bx = 416$, 
        $\bone^T \bx = 198$ and  
        $\bone^T \bd = 40$.   
So 
$(\bd^T \bx) n  - (\bone^T \bd)(\bone^T \bx) = -16 < 0$ and hence (\ref{eq:genfpglobal}) fails.

\begin{figure}[H]

\caption{Counterexample for the generalized global friendship inequality (\ref{eq:genfpglobal}) with $\ell = 2$ walk centrality. \label{fig:counter}}
\centering
\begin{tikzpicture}[
scale = 0.8,
    every node/.style={circle,draw,inner sep=2pt},
    >=stealth
]

\node (v1) at (0,0) {1};

\foreach \i [count=\n from 0] in {2,...,13} {
    \node (v\i) at (\n*30:3) {\i};
    \draw (v1) -- (v\i); 
}

\node (v14) at (6,1.5) {14};
\node (v15) at (7.5,3) {15};
\node (v16) at (9,1.5) {16};

\draw (v14) -- (v15);
\draw (v15) -- (v16);
\draw (v16) -- (v14);

\node (v17) at (6,-1.5) {17};
\node (v18) at (7.5,-3) {18};
\node (v19) at (9,-1.5) {19};

\draw (v17) -- (v18);
\draw (v18) -- (v19);
\draw (v19) -- (v17);

\draw (v15) -- (v17);
\draw (v14) -- (v2);

\end{tikzpicture}
\end{figure}
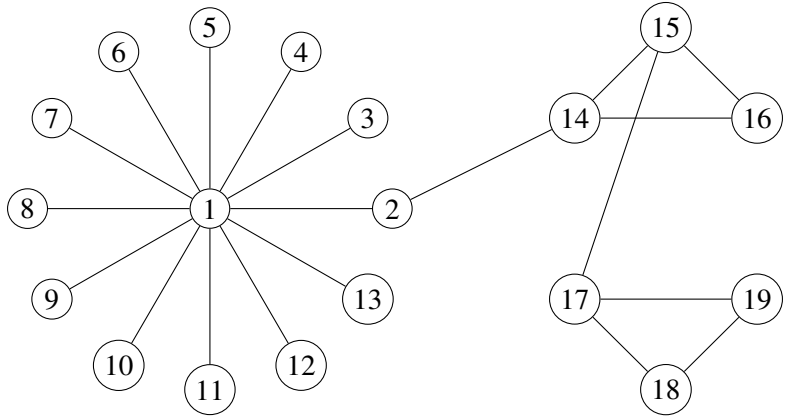

For the adjacency matrix associated with the graph in Figure~\ref{fig:counter} the two eigenvalues of largest absolute value are
$\lambda_1 \approx  3.4819$    
and $\lambda_n \approx -3.4793 $.
Because there is a single dominant eigenvalue,
as $\ell \to \infty$ the walk count centrality
$\bx = A^{\ell} \bone$ aligns with the 
Perron-Frobenius eigenvector. So, by 
Corollary~\ref{cor:fA}, for sufficiently large 
$\ell$ the inequality (\ref{eq:genfpglobal}) must hold.
This effect is illustrated 
in Figure~\ref{fig:even}, where the red  asterisks  
show the left hand side of (\ref{eq:genfpglobal}) 
with $\bx = A^{\ell}  \bone / \| A^{\ell}  \bone \|_2$
for $\ell = 2,4,6,\ldots,100$.
(Note from Remark~\ref{rem:scaling} that this normalization is justified.)
The two extreme eigenvalues have quite similar magnitudes, so a relatively large value of $\ell = 74 $ is needed for the effect of the dominant eigenvalue to 
lead to a positive left hand side.

\begin{figure}[htp]
    \centering
    \includegraphics[width=0.75\textwidth]{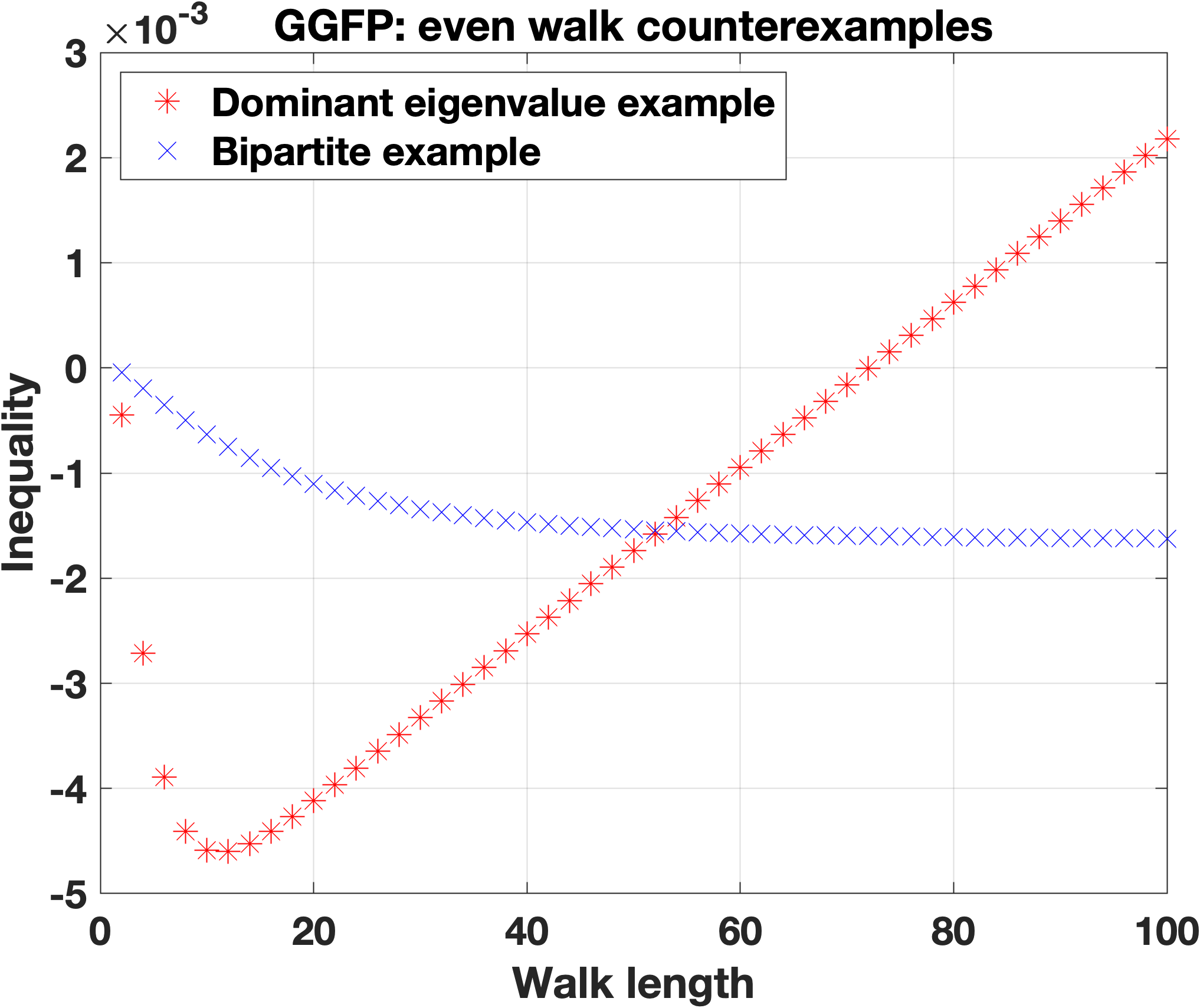}
    \caption{Left hand side of the 
    generalized global friendship inequality 
     (\ref{eq:genfpglobal})
     for even values of walk length
       $\ell$ in the normalized centrality measure
       $\bx = A^{\ell}  \bone / \| A^{\ell}  \bone \|_2$. Red stars correspond to the 
       graph in Figure~\ref{fig:counter},
       which provides a counterexample for
       even walk lengths $\ell = 2,4,6,\ldots,72$. 
       Blue crosses correspond to the bipartite
       graph in \cite[Figure~2]{Han2007}, which 
       provides a 
      counterexample for all even walk lengths. 
    }
    \label{fig:even}
\end{figure}

The blue crosses in Figure~\ref{fig:even} show the same values for the graph in 
\cite[Figure~2]{Han2007}. Here, because the graph is bipartite, there is a pair of dominant
eigenvalues, with opposite sign,  which we denote $\lambda_1 >0$
and $\lambda_n <0$.
We denote the corresponding eigenvectors  
$\bv_1$ and $\bv_n$.
It follows that as the even walk length $\ell$ tends to infinity the centrality vector
$\bx = A^{\ell}  \bone$ aligns with 
the vector $(\bone^T \bv_1) \bv_1 + 
(\bone^T \bv_n) \bv_n$. This limiting vector 
was found to violate the inequality
(\ref{eq:genfpglobal}), which explains why the
inequality fails to hold for large even walk lengths.

\begin{remark}\label{rem:r3}
It follows from Corollary~\ref{cor:fA}
that a generalized global friendship paradox holds
for $\bx = f(A) \bone$ when $f(A)$ is a convergent power series with only odd powers.
The counterexamples above show that a general 
statement of this form cannot be made when even powers are present.
\end{remark}

\subsection{Decreasing $f$ in $f(A) \bone$}\label{subsec:fdecrease}

The proof of Theorem~\ref{thm:globfA} makes use of the implication 
\[ f \text{ \ increasing} \implies (\lambda_i - \lambda_j)(f(\lambda_i)-f(\lambda_j)) \ge 0.\]
If $f$ is decreasing
over the spectrum of $A$, the reverse inequality holds; that is, 
\[ f \text{ \ decreasing} \implies (\lambda_i - \lambda_j)(f(\lambda_i)-f(\lambda_j)) \le 0,\]
and an analogue of Theorem~\ref{thm:globfA} follows with the inequality reversed. 
In this case, however, to have a valid 
centrality measure we must ensure that 
$\bx = f(A)\bone$
is nonnegative in all entries.

One example of a decreasing function
is $f(x) = \exp(-\beta x)$ for $\beta > 0$, which is used in \cite{EHH08}.
A closely related 
example is
$f(x) = 1/(1+\alpha x)$,
which is decreasing 
over the spectrum of $A$ 
for $0 < \alpha < 1/|\lambda_1|$, where  $\lambda_1 < 0$ is the most negative eigenvalue of $A$.
To obtain a condition for the 
nonnegativity of $\bx$ in this case 
we start with  
\[ \bx = (I+\alpha A)^{-1}\bone =\sum_{k=0}^{\infty}(-\alpha)^kA^k\bone,\]
which implies that
\[ \|\bx\|_{\infty} \le \sum_{k=0}^{\infty}(\alpha)^k(d_{\max})^k,\]
where $d_{\max}$ is the maximum degree. If we assume that $\alpha < 1/d_{\max}$ then the geometric sum converges, and 
$ \|\bx\|_{\infty} \le 1/(1-\alpha \, d_{\max})$.
We may then write 
\[ x_i = 1 - \sum_{j} \alpha \, a_{ij} x_j
\ge   1 - d_{\max} \, \alpha\, \|\bx\|_{\infty} \ge 1 - \frac{\alpha \, d_{\max}}{1-\alpha \, d_{\max}}.\]
The last quantity is nonnegative if 
$ \alpha \le 1/(2 d_{\max}) $.

\subsection{Functions of Loneliness}\label{subsec:fd} 
Using an argument similar to that in the proof of Lemma~\ref{lem:mon}, it follows
that Theorem~\ref{thm:revgfp} extends to any 
$\bx = f(\bd^{-1})$, where $f$ is an 
increasing function.

Using a similar strategy, we can
also extend Theorem~\ref{thm:revlfp} in the same way.
 In this case we may write  
\[ \bone^T(D^{-1}A \bx- \bx) = \sum_{i,j} \left(\frac{a_{ij}}{d_i}f\left(\frac{1}{d_j}\right)\right) -\left(\sum_if\left(\frac{1}{d_i}\right)\right) = \sum_{i,j}\left(\frac{a_{ij}}{d_i}f\left(\frac{1}{d_j}\right) -\frac{a_{ij}}{d_i}f\left(\frac{1}{d_i}\right)\right) =\]
\[ = \sum_{i,j}a_{ij}\left(\frac{1}{d_i}f\left(\frac{1}{d_j}\right) -\frac{1}{d_i}f\left(\frac{1}{d_i}\right)\right).\]
Since $A$ is symmetric we can rewrite this expression as
\[\frac{1}{2}\sum_{i,j}a_{ij}\left(\frac{1}{d_i}f\left(\frac{1}{d_j}\right)+\frac{1}{d_j}f\left(\frac{1}{d_i}\right) -\frac{1}{d_i}f\left(\frac{1}{d_i}\right)-\frac{1}{d_j}f\left(\frac{1}{d_j}\right)\right) \le 0,\]
which follows from the rearrangement inequality.

\subsection{Subgraph Centrality}\label{subsec:subgraph}

Closely related to walk-based centrality and 
exponential centrality is the 
\emph{subgraph centrality} measure
\cite{ER05}; see also \cite{BK13,EHSiamRev}.
This is defined as $\bx$ where $x_i = \exp(A)_{ii}$.
 After a computer search we found the adjacency matrix
\[
A = 
\left[
\begin{array}{ccccccc}
    0 & 1 &  0 &  0 &  1 & 1 & 1\\
    1  & 0 & 1  & 1 &  0 &  0 &  0\\
    0 &  1 &  0 & 1 & 0 & 0 & 1\\
    0  & 1 &  1 &  0 &  0 & 0 &  1\\
    1  & 0 &  0 &  0 & 0 & 0 & 0\\
    1 & 0 & 0 &  0 & 0 & 0 & 0\\
    1 & 0 & 1 &  1 &  0 & 0 & 0
    \end{array}
    \right]
\]
for which $\bx = (A^3)_{ii}$ produces
$ 
\bone^T( D^{-1} A \bx - \bx) = -1/3$. So the generalized local friendship inequality
(\ref{eq:genfplocal}) does not hold in general for $(f(A))_{ii}$ where 
$f(A)$ is a power series with positive coefficients, or where $f$ is increasing on the spectrum of $A$. However, 
it appears to be an open question whether a paradox holds for 
$\exp(A)_{ii}$, or for the diagonal entries of any other specific 
convergent power series.

\subsection{The case of $\lambda = 1$ for Nonbacktracking Eigenvector Centrality} \label{subsec:l1}

For the nonbacktracking eigenvector centrality
measure 
introduced in subsection~\ref{subsec:ggfpnb},
the special case where the dominant real eigenvalue $\lambda = 1$ leads to equality
in both friendship paradoxes. This motivates the 
discussion in this subsection, where we characterize when this special case arises. 

We first define the nonbacktracking or Hashimoto matrix, $B$ \cite{MZN14,NQ24,ST96}. Suppose the graph has $m$ edges.
We will regard each edge $i - j$ as a reciprocal 
pair of directed edges, $i \to j$ and $j \to i$.
Then, for a given ordering of these edges, 
$B \in \RR^{2m \times 2m}$ has the form
\[
b_{ i \to p, q \to j}
=
\left\{
\begin{array}{ll}
 1 & \text{~if~} i \neq j \text{~and~} p = q,
\\
 0 & \text{~otherwise.} 
 \end{array}
\right.
\]
In this way, $B$ records the presence of nonbacktracking walks of length two.
It is known, see for example \cite[Eq. (7)]{GK21}\cite{MZN14}, that if $\bv$ is an eigenvector of $B$  associated with the eigenvalue $\lambda \neq \pm 1$, then 
a corresponding eigenvector of the system
(\ref{eq:nbteig}) 
is given by 
\[
\left[
\begin{array}{c}
    \bx \\
    \by
\end{array}
\right],
\quad \text{~with~} \quad
 y_i = \sum_{i \leftrightarrow j} v_{j \rightarrow i} \quad x_i = \lambda y_i.
 \]

On the other hand, if $\begin{bmatrix}
    \bx \\ \by
\end{bmatrix}$ is an eigenvector of the system (\ref{eq:nbteig}), then 
\[ v(i \rightarrow j)  =  y_j -\frac{1}{\lambda} y_i\]
is an eigenvector of $B$, where we set $\frac{1}{\lambda} = 0$ if $\lambda = 0$. 
 Furthermore, we have a bijection between eigenvalues of the two matrices, excluding $\lambda = \pm 1$. This result is known as Ihara's formula.
 \begin{theorem}[{\cite[Thm. 2.1]{GK21}}]
 For a graph with $n$ vertexs and $m$ edges
 we have 
\[\det(I - uB) = (1 - u^2)^{m-n}\det(u^2(D - I) - uA + I).\]
 \end{theorem}

 As mentioned in subsection~\ref{subsec:ggfpnb}, the system (\ref{eq:nbteig}) always has $1$ as eigenvalue, with eigenvector $\bone$.
 If $\lambda = 1$ is dominant we use this  
 uniform 
 eigenvector to define the centrality measure.
 (The matrix $B$ may have $1$ as a multiple eigenvalue (multiple cycles) or may not have $1$ as an eigenvalue (for trees, $B$ is nilpotent).)

To fully characaterize the cases where 
the nonbacktracking matrix spectral radius strictly greater than $1$, we use some results from \cite{GK21}.
\begin{lemma}[{\cite[Prop. 5.5]{GK21}}]
Let $G$ be a connected graph such that $G$ is not a tree or cycle and with minimum degree $d_{\min} \ge 2$. Then the spectral radius of the nonbacktracking matrix is strictly greater than $1$.
\end{lemma}

The condition $d_{\min} \ge 2$ tells us that the graph has no dangling nodes. This assumption can be relaxed thanks to the following lemma.

\begin{lemma}[{\cite[Corollary 3.2.1]{GK21}}]\label{lem:GK}
    Let $G$ be a graph with $m$ edges and $T$ be a tree with $n$ vertices. Let $B$ be
the nonbacktracking matrix of $G$. Define $\widehat{G}$ as the graph constructed by joining $G$ and $T$ on one
vertex. Define $\widehat{B}$ as the nonbacktracking matrix of $\widehat{G}$. Then 
the eigenvalues of $\widehat{B}$ 
are precisely those of $B$ along with
an eigenvalue $0$ with algebraic multiplicity $2(n - 1)$.
\end{lemma}

Lemma~\ref{lem:GK} tells us that adding or removing dangling nodes does not alter the spectral radius of the nonbacktracking matrix. As a consequence, the nonbacktracking matrix of the $2$-core of a graph has the same spectral radius as the original nonbacktracking matrix. 

In conclusion, we have the following theorem.

\begin{theorem}\label{thm:twocore}
    The spectral radius of the nonbacktracking matrix is strictly greater than $1$ if and only if the $2$-core of $G$ is not a cycle (or empty).
\end{theorem}

\subsection{Well Posedness of Nonbacktracking Eigenvector Centrality}\label{subsec:nnbtheory}

In this subsection, for completeness we show that 
for the eigenvalue problem (\ref{eq:nbteig}) it is always possible to construct a 
nonnegative eigenvector associated with a dominant real eigenvalue.
Using the relationship mentioned in subsection~\ref{subsec:l1}, we will do this by studying the closely related eigenvalue problem involving $B$.
The vector that is constructed may then be used to define nonbacktracking eigenvector centrality.
To proceed, we would like to invoke the Perron-Frobenius theorem. However, we must address the issue that in certain cases the matrix $B$ is not irreducible. 

From \cite[Prop. 2.3]{GK21} we have the following result.

\begin{theorem}[Irreducibility of $B$ \cite{GK21}]
\label{thm:GK}
    Let $G$ be a connected graph that is not a cycle and $d_{\min} \geq 2$ (i.e. no dangling nodes). Then $B$ is
irreducible.
\end{theorem}

Theorem~\ref{thm:GK} lists generic cases where nonbacktracking centrality is immediately well-defined.  
We must analyze three remaining cases:
\begin{itemize}
    \item The graph is a tree.
    \item The graph is a cycle with some attached dangling nodes (i.e., the $2$-core of the graph is a cycle).
    \item The $2$-core contains at least two cycles, but the graph has dangling nodes.
\end{itemize}

In all these cases, the matrix $B$ is reducible. In fact, let $v$ be a dangling node and let $u$ be its neighbour. Then the row associated with $u \rightarrow v$ in $B$ contains only zeros. Furthermore, if $G$ is a cycle then $B$ 
has the form
\[ B = \begin{bmatrix}
    P & 0 \\
    0 & P
\end{bmatrix},\]
where $P$ is a permutation matrix. Clearly such a matrix is reducible. 


Our aim is therefore to show that a well-defined nonbacktracking eigenvector centrality vector exists in the three cases above.

First, assume that $G$ is a tree or a cycle with possibly dangling nodes attached. Then from 
Theorem~\ref{thm:twocore} and the proof of 
Theorem~\ref{thm:globnbt} we know that the  
block matrix in (\ref{eq:nbteig}) has spectral radius equal to $1$, with the 
eigenvector $\bone$ corresponding to the 
eigenvalue $1$. We may use this eigenvector to 
define nonbacktracking centrality. 

Now assume that $G$ has a $2$-core containing at least two cycles, but $G$ has dangling nodes. We start with the leading eigenvalue and positive eigenvector of the matrix $B$ associated with the $2$-core ($B$ is irreducible in this case and we can appeal to the Perron-Frobenius Theorem). If we add a dangling node $v$ to an existing node $u$, then the matrix $B$ is enlarged by two in each dimension. Using ``old'' and ``new''
to denote the matrices before and after this addition, we have 
\[ B_{\mathrm{new}} = \begin{bmatrix}
    B_{\mathrm{old}} & 0_{2m \times 1} & \br\\
    \br^T &0 &0\\
    0_{2m \times 1}& 0 & 0
\end{bmatrix},\]
where each entry of $\br \in \mathbb{R}^{2m}$ is either zero or one, and the last two rows and columns are associated with $v \rightarrow u$ and $u \rightarrow v$, respectively. We will show that this matrix has a nonnegative eigenvector even though it is reducible.
Consider the vector 
\[ \bx_{\mathrm{new}} = \begin{bmatrix}
    \bx_{\mathrm{old}} \\ a \\ 0
\end{bmatrix} \in \RR^{2m+2},\]
where $\bx_{\mathrm{old}} \in \mathbb{R}^{2m}$ 
is the leading eigenvector of $B_{\mathrm{old}}$,
which 
has nonnegative entries, and $a \in \mathbb{R}$ is to be determined. 
Since $ B \bx_{\mathrm{old}} = \rho(B_{\mathrm{old}}) \bx_{\mathrm{old}}$,
 we have constructed a dominant eigenvector of 
$B_{\mathrm{new}}$ if $a$ satisfies
\[
        \br^T \bx_{\mathrm{old}} = \rho(B_{\mathrm{old}})a,
\]
which shows that a suitable nonnegative $a$ may be found.

We see that after adding a dangling node 
it remains possible to construct a 
nonnegative leading eigenvector associated with a positive real eigenvalue of maximum modulus. In this way, we may add all the required dangling nodes to the $2$-core until we have recovered the original graph, and the 
resulting
nonnegative leading eigenvector can be used to 
define 
nonbacktracking eigenvector centrality.

\section{Computational Experiments}\label{sec:exp}

In this section we report on computational experiments that record the size of the generalized friendship effects. To do this  
we define the \emph{global friendship paradox margin} to be  $(\bd^T \bx)/(\bone^T \bd) - 
(\bone^T \bx)/n$, 
and the 
\emph{local friendship paradox margin} to be 
$\bone^T (D^{-1} A \bx - \bx)/n$. In the later case we
have normalized by $n$ to account for the network size.
From the perspective of social network analysis, it is interesting to compare these margins across different network types and different centrality measures. We may also ask whether
the margin size for one centrality measure always dominates the margin size for another.

We consider four real-world undirected networks and three synthetic random graph models, all sourced from the NetworkX Python library \cite{hagberg2020networkx}.

The \textbf{Karate Club} network, from Zachary's seminal 1977 study \cite{zachary1977information}, represents friendships among 34 members of a university karate club in the United States, which famously split into two factions following a dispute between the instructor and club president. The \textbf{Davis Southern Women} network, collected by Davis et al.\ in the 1930s \cite{davis1941deep}, is a bipartite graph representing the observed attendance of 18 Southern women at 14 social events. The \textbf{Florentine Families} network, from Padgett and Ansell's 1993 study \cite{padgett1993robust}, captures marriage alliances among 15 prominent Florentine families during the Italian Renaissance (circa 1430), a period marked by political rivalry between the Medici and Strozzi factions. The \textbf{Les Mis\'{e}rables} network, compiled by Knuth \cite{knuth1993stanford}, represents co-appearances of 77 characters in Victor Hugo's novel, where edges connect characters who appear in the same scene. 

The three synthetic random graph models are generated with $n=50$ nodes. The \textbf{Erd\H{o}s-R\'{e}nyi},
or Gilbert, 
 model $G(N, p)$ \cite{erdos1959random,Gil59} 
  constructs a graph where each possible edge is included independently with probability $p=0.2$. The \textbf{Newman-Watts-Strogatz} model \cite{NEWMAN1999341} generates small-world networks by starting from a ring lattice where each node is connected to its $k=3$ nearest neighbours on each side, then adding shortcut edges with probability $p=0.1$ for each existing edge.
   Unlike the original Watts-Strogatz model, which rewires edges, this variant preserves all original connections while adding new shortcuts, producing connected graphs that exhibit high clustering (like regular lattices) while maintaining short average path lengths (like random graphs).
The \textbf{Random Geometric} model \cite{penrose2003random} places nodes uniformly at random in the unit square and connects pairs whose Euclidean distance is at most the radius $r=0.3$, creating spatially-embedded networks with inherent locality and clustering properties characteristic of wireless communication networks and other physical systems.

Figure~\ref{fig:real_and_random_graphs_margins}
shows results for degree, Katz, exponential, and nonbacktracking eigenvalue centralities. 
In all computations, we normalize $\bx$ so that 
$\| \bx \|_1 = 1$. 
Results for the synthetic random graphs are averaged over 10 trials to account for variability in the random generation process.

We see that the four centrality measures typically produce margins that are within an order of magnitude of each other, and that the relative ordering of the margins is not consistent across different networks.

\begin{figure}
\centering
\includegraphics[width=1\textwidth]{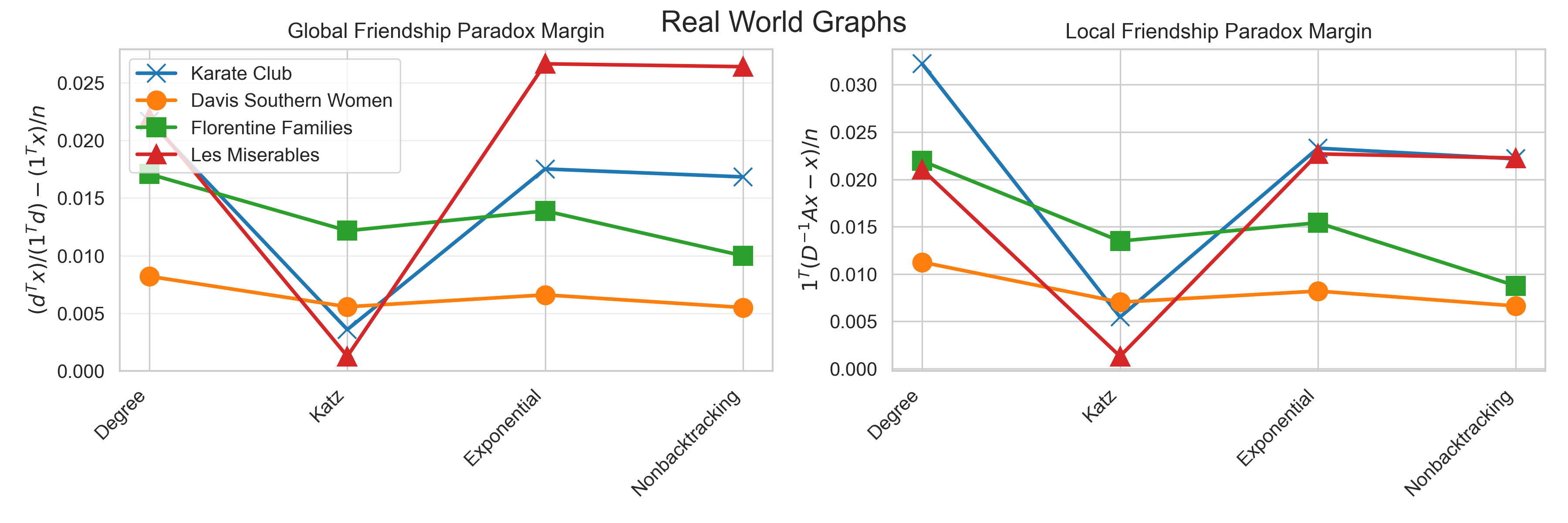}\\
\includegraphics[width=1\textwidth]{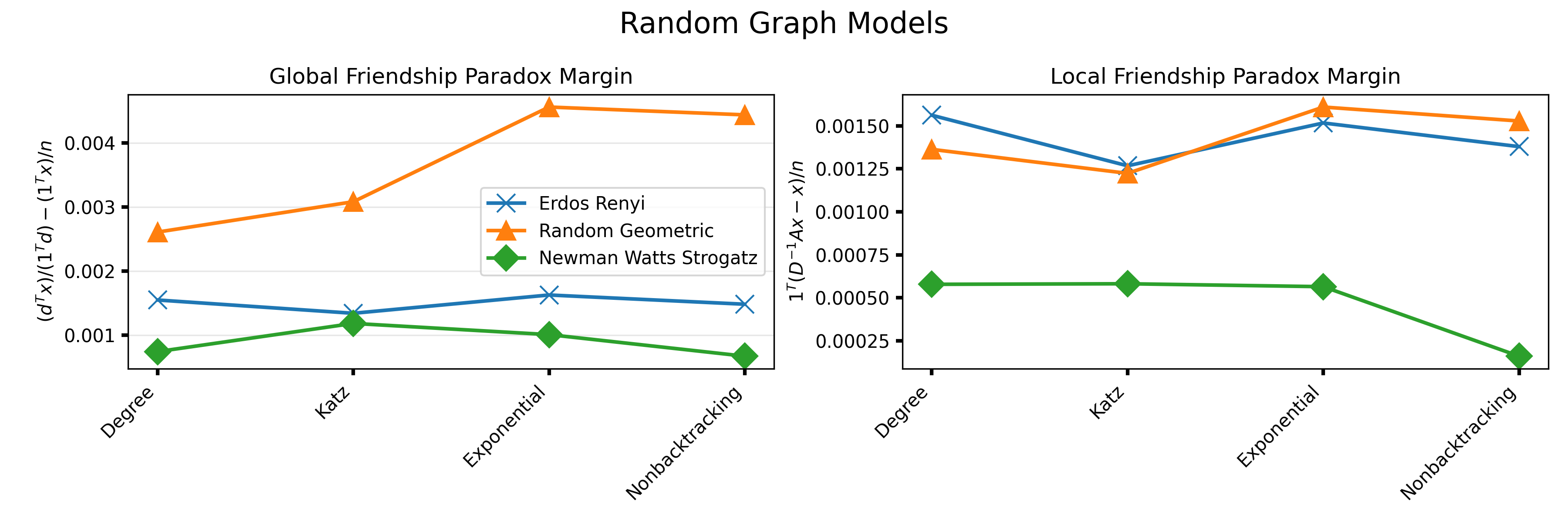}
\caption{Global and local friendship paradox margins for different centrality measures on real-world networks (top) and random graph models (bottom). The global friendship paradox margin is computed as $(\bd^T \bx)/(\bone^T \bd) - (\bone^T \bx)/n$, while the local friendship paradox margin is computed as $\bone^T (D^{-1} A \bx - \bx)/n$. Results for the  random graphs are averaged over 10 trials.}\label{fig:real_and_random_graphs_margins}
\end{figure}

Additional results are  presented in Tables~\ref{tab:global_inequality} and \ref{tab:local_inequality}. The tables show results for larger real-world networks spanning diverse domains. The \textbf{USAir97} air transportation network from \cite{DH11}, with weights binarized, consists of $N=332$ nodes representing US airports connected by 2126 undirected edges indicating whether at least one scheduled USAir flight operated between airport pairs in 1997. The \textbf{Celegans} neural connectivity network contains 277 nodes representing neurons and 2105 edges representing synaptic connections in the nematode \textit{Caenorhabditis elegans}, sourced from \verb5https://www.cs.cornell.edu/~arb/data/5 based on \cite{GD03}. The \textbf{Jazz} collaboration network from \cite{CMK04} comprises 198 jazz bands as nodes connected by 2742 edges representing shared musicians who performed together between 1912 and 1940. The \textbf{Adjnoun} network from \cite{Newman06} represents word adjacency patterns in the novel ``David Copperfield'' by Charles Dickens, where the 112 nodes correspond to the 60 most commonly occurring adjectives and 60 most commonly occurring nouns (excluding 8 disconnected nodes), connected by 425 edges indicating adjacent positions in the text. The \textbf{Football} network from \cite{girvan2002community} represents American football games among 115 Division IA college teams during the fall 2000 regular season, where nodes represent teams and the 613 edges represent scheduled matches. Finally, the \textbf{Journals} readership network from \cite{pajek_data} contains 124 nodes representing magazines and journals connected by 5972 edges indicating shared reader interests, derived from a survey of approximately 100,000 Ljubljana (Slovenia) residents conducted in 1999--2000.

\begin{table}[htbp]
\centering
\caption{Global Friendship Paradox Margin: $(\bd^T \bx)/(\bone^T \bd) - (\bone^T \bx)/n$}
\label{tab:global_inequality}
\begin{tabular}{lcccccc}
\toprule
Centrality & Adjnoun & Football & Jazz & USAir97 & Journals & C.elegans \\
\midrule
Degree & 7.2763e-03 & 5.9727e-05 & 1.9957e-03 & 7.4213e-03 & 4.4923e-04 & 2.6964e-03 \\
Katz & 5.5946e-03 & 8.0152e-05 & 1.7464e-03 & 4.5260e-03 & 3.3526e-04 & 1.8515e-03 \\
Exponential & 6.5420e-03 & 9.2335e-05 & 2.2484e-03 & 6.6854e-03 & 3.8577e-04 & 2.3047e-03 \\
Nonbacktracking & 5.8218e-03 & 8.9102e-05 & 2.2102e-03 & 6.5785e-03 & 3.8286e-04 & 2.1628e-03 \\
\bottomrule
\end{tabular}
\end{table}

\begin{table}[htbp]
\centering
\caption{Local Friendship Paradox Margin: $1^T (D^{-1} A \bx - \bx)/n$}
\label{tab:local_inequality}
\begin{tabular}{lcccccc}
\toprule
Centrality & Adjnoun & Football & Jazz & USAir97 & Journals & C.elegans \\
\midrule
Degree & 8.4392e-03 & 4.9844e-05 & 1.6975e-03 & 8.7216e-03 & 5.4808e-04 & 4.0817e-03 \\
Katz & 6.0686e-03 & 4.9344e-05 & 1.2739e-03 & 4.5896e-03 & 4.0626e-04 & 2.3374e-03 \\
Exponential & 7.0473e-03 & 5.4593e-05 & 1.5481e-03 & 6.5988e-03 & 4.6692e-04 & 2.7848e-03 \\
Nonbacktracking & 6.1646e-03 & 5.0445e-05 & 1.5136e-03 & 6.4232e-03 & 4.6328e-04 & 2.5152e-03 \\
\bottomrule
\end{tabular}
\end{table}

\section{Summary}\label{sec:summary}

Our main aim was to present new results on friendship paradoxes---inequalities that can be shown to hold for all connected, undirected networks. 
To do this, we set out a unified treatment of the global/local and degree/generalized versions of the paradox from a straightforward linear algebra perspective.
The original observation in this field dates back more than thirty years \cite{Feld91}, subsequently initiating many 
empirical experiments and motivating researchers across several disciplines to consider interpretations, repercussions and the potential for 
algorithmic developments. 
It is perhaps surprising, therefore, that rigorous
extensions have been discovered only quite recently
\cite{CKN21,HV25,H19} and that further results are possible.

For convenience, 
in Table~\ref{tab:summary} we summarize the 
current known results.
The table clarifies that some open questions remain.
Is there an analogue of Theorem~\ref{thm:globfA}
for the local paradox?
Can results, or counterexamples be found for subgraph centrality, or other measures of the form $(f(A))_{ii}$?
Equally intriguingly, we know that nonbacktracking eigenvector centrality,
which is covered by 
Theorems ~\ref{thm:globnbt} and 
\ref{thm:locnbt}, 
arises as the limit of Katz-style nonbacktracking walk centrality when the downweighting parameter approaches its upper limit \cite{GHN18}.
Can paradoxes, or counterexamples, be found for nonbacktracking Katz?

\begin{table}[t] 
\caption{Summary of Friendship Paradox Results for Connected, Undirected Networks\label{tab:summary}}
\begin{tabular}{ l | cc }
  \toprule
Attribute   & Global Paradox & Local Paradox\\
\midrule
       Degree & \cite{Feld91} & \cite{CKN21} \\
Eigenvector                & \cite{H19} & \cite{HV25} \\ 
$A^{\ell} \bone $, for $\ell$ odd & Corollary~\ref{cor:fA}   & \cite{HV25} (equality case in Theorem~\ref{thm:even}) \\
$A^{\ell} \bone $, for $\ell$ even & counterexample in Sec.~\ref{subsec:counterggfi}  & \cite{HV25}
(equality case in Theorem~\ref{thm:even})
\\
Katz & Corollary~\ref{cor:fA} & \cite{HV25} 
\\
General $f(A) \bone$, $f$ power series  & Remark~\ref{rem:r3}     & Remark~\ref{remark:r2} \\  
Exponential  & Corollary~\ref{cor:fA} & 
Remark~\ref{remark:r1} (also follows from  \cite{HV25})    \\
General $f(A) \bone$, $f$ increasing & Theorem~\ref{thm:globfA} & not known \\
Loneliness & Theorem~\ref{thm:revgfp} & 
Theorem~\ref{thm:revlfp}\\ 
Nonbacktracking eigenvector & 
Theorem~\ref{thm:globnbt} & Theorem~\ref{thm:locnbt}\\
Geometric mean degree & Theorem~\ref{thm:globgmfp} &  Theorem~\ref{thm:locgmfp}\\
Subgraph: $(\exp(A))_{ii}$ & not known & not known\\
Nonbacktracking walk-based 
 & not known & not known\\
 \bottomrule
\end{tabular}
\end{table}

\section*{Acknowledgement} 
DJH was supported by a Fellowship from the Leverhulme Trust
 and by the Advanced Grant ``Numerical Analysis for Stable AI''
  101198795 from the European Research Council.
FH was supported by a Mathematical Institute Scholarship by Oxford Mathematical Institute. 
FT was partially funded by the PRIN-MUR project MOLE code 2022ZK5ME7 and by PRIN-PNRR project FIN4GEO CUP P2022BNB97.

\bibliographystyle{siam}
\bibliography{fofrefs}

\end{document}